\documentclass{amsart}
\usepackage{amsmath}
\usepackage{amsfonts}

\theoremstyle{plain}

\newtheorem{definition}{Definition}

\newtheorem{lemma}{Lemma}

\newtheorem{remark}{Remark}

\newtheorem{theorem}{Theorem}
\numberwithin{equation}{section}

\begin{document}

\title[EMpirical characterization for the extremes]{How many statistics are needed to characterize the univariate extremes}
\author[G. S. LO]{Gane Samb LO}
\address{LSTA, UPMC, France and LERSTAD, Universit\'e Gaston Berger
de Saint-Louis, SENEGAL}
\email{gane-samb.lo@ugb.edu.sn, ganesamblo@ufrsat.org}
\urladdr{www.lsta.upmc.fr/gslo}

\keywords{Extreme value theory, order statistics, empirical distribution function,
domain of attraction of the maximum, functional and empirical characterization.\\ \\
This paper is part of the Doctorate of Sciences of the author, Dakar University 1991, under the title :
\textit{Empirical characterization of the extremes : A family of characterizing statistics}.}

\begin{abstract}
\Large
Let $X_{1},X_{2},...$ be a sequence of independent random variables ($rv$) with common distribution function ($df$) $F$ such that $F(1)=0$. We consider the simple statistical problem : find a statistics family of size $m\geq 1$ whose convergence, in probability or almost surely, to a point of some domain $\mathcal{S} \in \mathbb{R}^{m}$ is equivalent that $F$ lies in the extremal domain of attraction $\Gamma$. Such a family, whenever it exists, is called an Empirical Characterizing Statistics Family for the EXTtremes (ECSFEXT). The departure point of this theory goes back to Mason \cite{mason} who proved that the Hill (\cite{hill}) estimator converges a.s. to a positive real number for some particular sequences if and only $F$ lies in the attaction domain of a Fréchet's law. Considered for the whole attraction domain, the question becomes more complex. We provide here an ECSFEXT of nine (9) elements and also characterize the subdomains of $\Gamma$. The question of lowering m=9 to a minimum number is launched. 
\end{abstract}

\maketitle

\section{Introduction and statement of the main problem}
Let $X_{1},X_{2},...$ be a sequence of independent and identically associated with the
$df$ $F$, with $F(1)=0$ and let for once $G(y)=F(e^{x})$ an auxilliary
$df$ associated with independent and identically distributed
random variables $\log X_{1},$ $\log X_{2},....$ For each $n\geq 1$ fixed,
their order statistics are denoted by

\begin{equation*}
X_{1,n}=\log Y_{1,n}\leq X_{2,n}=\log Y_{2,n}\leq ...\leq X_{n,n}=\log Y_{n,n}.
\end{equation*}

\noindent The departure problem of Univariate Extreme Value Theory (UEVT) is finding the asymptotic law of the maximum observation
$X_{n,n}=\max \left( X_{1},...X_{n}\right)$. In this theory, the $df$ $F$ is said to be attracted to a non degenerated extremal $df$ $M$
iff the maximum $X_{n,n}=\max \left( X_{1},...X_{n}\right)$, when
appropriately centered and normalized \ by two sequences of real numbers $%
\left( a_{n}>0\right) _{n\geq 0}$ and $\left( b_{n}\right) _{n\geq 0}$,
converges to $M$, in the sense that

\begin{equation}
\lim_{n\rightarrow +\infty }\mathbb{P}\left( X_{n,n}\leq a_{n}\text{ }%
x+b_{n}\right) =\lim_{n\rightarrow +\infty }F^{n}\left( a_{n}x+b_{n}\right)
=M(x),  \label{dl05}
\end{equation}
for continuity points $x$ of $M$. If (\ref{dl05}) holds, it is said that $F$
is attracted to $M$ or $F$ belongs to the domain of attraction of $M$,
written $F\in D(M)$. It is well-kwown that the three nondegenerate possible
limits in (\ref{dl05}), called extremal $df$'s, correspond to three possibles are the following. \newline

\noindent The Gumbel $df$ 
\begin{equation}
\Lambda (x)=\exp (-\exp (-x)),\text{ }x\in \mathbb{R},  \label{dl05a}
\end{equation}
the Fr\'{e}chet $df$ with parameter $\gamma >0$

\begin{equation}
\phi _{\gamma }(x)=\exp (-x^{-\gamma })\mathbb{I}_{\left[ 0,+\infty \right[
}(x),\text{ }x\in \mathbb{R}\   \label{dl05b}
\end{equation}
and the Weiblull $df$ with parameter $\beta >0$

\begin{equation}
\psi _{\beta}(x)=\exp (-(x)^{\beta})\mathbb{I}_{\left] -\infty ,0\right]
}(x)+(1-1_{\left] -\infty ,0\right] }(x)),\ \ x\in \mathbb{R},\ 
\label{dl05c}
\end{equation}
where $\mathbb{I}_{A}$ denotes the indicator function of the set A.\\

\noindent Actually the limiting $df$ $M$ is defined by an equivalence class
of the binary relation $\mathcal{R}$ on the set $\mathcal{D}$ of $%
cdf^{\prime}s$ on $\mathbb{R}$, defined as follows 
\begin{equation*}
\forall (M_{1},M_{2})\in \mathcal{D}^{2},(M_{1}\text{ }\mathcal{R}\text{ }%
M_{2})\Leftrightarrow \exists (a,b)\in \mathbb{R}_{+}\backslash \{0\}\times 
\mathbb{R},\forall (x\in \mathbb{R}) :
\end{equation*}
\begin{equation*}
M_{2}(x)=M_{1}(ax+b).
\end{equation*}
One easily checks that if $F^{n}\left( a_{n}x+b_{n}\right) \rightarrow
M_{1}(x),$ then $F^{n}\left( c_{n}x+d_{n}\right) \rightarrow
M_{1}(ax+b)=M_{2}(x)$ whenever 
\begin{equation}
a_{n}/d_{n}\rightarrow a\text{ and }(b_{n}-d_{n})/c_{n}\rightarrow b\text{
as }n\rightarrow \infty .  \label{dl05f}
\end{equation}

\noindent Theses facts allow to parameterize the class of extremal $df$'s.
For this purpose, suppose that (\ref{dl05}) holds for the three $df$'s
given in (\ref{dl05a}), (\ref{dl05b}) and (\ref{dl05c}). If we take
sequences $(c_{n}>0)_{n\geq 1}$ and $(d_{n})_{n\geq 1}$ such that the limits
in (\ref{dl05f}) are $a=\gamma=1/\alpha$ and $b=1$ (in the case of Fr\'{e}%
chet extremal domain), and $a=-\beta=-1/\alpha$ and $b=-1$ (in the case of
Weibull extremal domain), and finally, if we interpret $(1+\gamma
x)^{-1/\gamma}$ as $exp(-x)$ for $\gamma =0$ (in the case of Gumbel extremal
domain), we are entitled to write the three extremal $df$'s in the
parameterized shape 
\begin{equation}
G_{\gamma }(x)=\exp (-(1+\gamma x)^{-1/\gamma }),\text{ }1+\gamma x\geq 0,
\label{dl05d}
\end{equation}
called the Generalized Extreme Value (GEV) distribution function with parameter $\gamma \in 
\mathbb{R}$.\\

\noindent For a for a modern and large account of the Extreme Value Theory, the reader is referred to Beirlant \textit{et al.} \cite{bgt}, Galambos \cite{galambos}, de Haan \cite{dehaan}, de Haan and Ferreira \cite{dehaan1} and Resnick \cite{resnick}.\\

\noindent The problem of estimating the extremal index $\gamma$ by various and numerous estimators and finding statistical tests based on those estimators has been extremely widely tackled by many authors in papers and books. Let us only cite a sample of these authors as : Pickands \cite{pickands}, Hall
(1981) \cite{hallf}, Berilant and Teugels (1986) \cite{berf}, Deheuvels and Mason (1990) \cite{dmf1}, Deheuvels and Mason (1990) \cite{dmf3}, Deheuvels,
Haeusler and Mason (1988) \cite{dhmf}, Cs\"{o}rg\H{o}, Haeusler and Mason \cite{chm} and L\^{o} \cite{gsloa}, \cite{gslob} etc. Even in these last years new statistics continue to appear in the frame of new methodologies such as adaptative procedures and second and third order condition, etc.\\

\noindent This paper is not only about statistical estimation of the extremal domain, in the sense that the convergence of some statistics $S_{n}$ to a function of the extremal index $g(\gamma)$, under the hypothesis (H) that $F$ lies in $\Gamma$, yields a statistical test of (H) with $\left(|S_{n}-g(\gamma))| > c\right)$ as a rejection region. We also face the inverse question : does the convergence of $S_{n}$ to $g(\gamma)$ implies that (H) holds. This is the empirical characterization problem that we set and motivate in the next Section \ref{sec2}. In Section \ref{sec3}, we give a general solution proved in Section \ref{sec4}. Concluding remarks end the paper in Section \ref{sec5}.

\section{The problem and its motivation} \label{sec2}

We are now describing the Mason fundamental result which is the departure point of our question. Consider
a sequence of integers $k=k_{n}$, $n\geq 1$ satisfying,%
\begin{equation}
k_{n}\rightarrow \infty \text{ and }k_{n}/n\rightarrow 0\text{ as }%
n\rightarrow \infty  \tag{K}
\end{equation}%
and consider for $0<\alpha <1,$ $k_{n}(\alpha )=[n^{\alpha }]$, where $[x]$
denotes the integer part of $x$, that is the unique integer such that $%
[x]\leq x<[x]+1$, special cases of sequences satisfying $(K)$. Mason \cite{mason} (1982)
proved the following.

\bigskip

\begin{theorem} \label{theo21} For any $0<\gamma <\infty ,$ and $\ell =1$, $F\in D(\phi _{\gamma })$ if and only if
\begin{itemize}

\item [(i)] for some $0<\alpha <1$, 
\begin{equation*}
T_{n}(2,k,\ell )=k^{-1}\sum_{i=\ell }^{i=k}i(\log X_{n-i+1,n}-\log
X_{n-i,n})\rightarrow 1/\gamma ,p.s.
\end{equation*}

\noindent as $n\rightarrow +\infty $
\item [(ii)] if and only if for all sequences satisfying $(K),$%
\begin{equation*}
T_{n}(2,k,\ell ) \rightarrow _{P}1/\gamma
\end{equation*}

\noindent as $n\rightarrow +\infty$.
\end{itemize}
\end{theorem}

\bigskip

\noindent This is the first step of what we call empirical characterizations of the
extremes achieved only with the Hill statistic $T_{n}(2,k,\ell )$. From this,
we formulate the following general problem.

\bigskip

\noindent Given only the observations $X_{1},X_{2},...$ associated with an unknow
underlying $df$ $F$, is it possible to answer these three
questions ? First

\bigskip

\noindent ($\mathcal{P}$) : Is it possible to find a set of statistics, that is a vector of $m \geq 1$
statistics $S_{n}=(S_{n}(1),...,S_{n}(m))$ and a subset $\mathcal{S}$ of $\mathbb{R}^{m}$ that such
the convergence of $S_{n}$ to a point of $\mathcal{S}$ is a necessary and sufficient
condition for $F$ to ly in the extremal domain $\Gamma$?

\newpage
\bigskip

\noindent This problem may be rephrased as follows : Is it possible to demonstrate the
existence of $(S_{n})$ and $\mathcal{S}$ such  that : 
\begin{equation}
(F\in \Gamma )\iff \exists (a\in \mathcal{S}),\text{ }(S_{n}\rightarrow a)  \label{c00}
\end{equation}%

\noindent where the limit is almost sure or in probability.\\

\bigskip

\noindent We denote this as a \textbf{global} empirical characterization of the extremal
domain. The statistics $S_{n},$ if it exists will be called an Empirical Characterizing Statistics Family for
the Extremes (ECSFEXT).\\

\noindent Let us introduce this notation. For $\mathcal{S}$ of $\mathbb{R}^{m}$, with $m \geq 1$, we call $\Pi$($\mathcal{S}$) the set all projections of $\mathcal{S}$ on its components.\\

\noindent If this question is positively answered, we go further and find to search to partition $\mathcal{S}$
into three subdomains $\mathcal{S}_{0}$, $\mathcal{S}_{1}$ and $\mathcal{S}_{2}$, the two latters being
paremeterized by $\gamma >0$, that this 
\begin{equation}
S_{1}=\{a(\gamma ),\gamma >0\},S_{2}=\{b(\gamma ),\gamma >0\}  \label{param}
\end{equation}
such that there exists $\pi \in \Pi(\mathcal{S})$ so that
\begin{equation}
(F\in \Lambda )\iff \exists (a\in S_{0}),\text{ }(\pi(S_{n})\rightarrow \pi(a)),
\label{c01}
\end{equation}%
for any $\gamma >0,$%
\begin{equation}
(F\in D(\phi ))\iff \left\{ \exists (\gamma >0),\text{ }(\pi(S_{n}) \rightarrow
\pi(a(\gamma)) \in S_{1})\text{ } \Rightarrow (F\in D(\phi _{\gamma })\right\}
\label{c02}
\end{equation}%
and for any%
\begin{equation}
(F\in D(\psi ))\iff \left\{ \exists (\gamma >0),\text{ }(\pi(S_{n})\rightarrow
\pi(b(\gamma))\in S_{2})\text{ } \Rightarrow (F\in D(\phi _{\gamma })\right\}.
\label{c03}
\end{equation}

\noindent When the empirical characterization concerns any particular
case (\ref{c01}), (\ref{c02}) or (\ref{c03}), we qualify it as \textbf{
simple}.\\

\noindent At this point, Mason \cite{mason} have solved the case (\ref{c02}) in a very general way,
both in probability limits and in almost sure limits.\\

\bigskip

\noindent We should not be confusing this empirical characterization problem with that of 
the estimation or the selection of the extremal domain. For this, we have :

\begin{definition}  A family of $m$ statistics $S_{n}$ is an Estimating Statistics Family for the Extremes (ESFEXT) if there exists a subset $\mathcal{S}$ of $\mathbb{R}^{m}$ partitionned into $\mathcal{S}_{0},$ $\mathcal{S}_{1}$ and $\mathcal{S}_{2}$ where $\mathcal{S}_{1}$ and $\mathcal{S}_{2}$ are parameterized as in (\ref{param}), such that%
\begin{equation}
(F\in \Lambda )\implies S_{n}\rightarrow a\in S_{0},
\end{equation}%
for any $\gamma >0,$%
\begin{equation}
(F\in D(\phi _{\gamma }))\implies (S_{n}\rightarrow a(\gamma )\in S_{1}
\end{equation}%
and for any $\gamma >0$%
\begin{equation}
(F\in D(\psi _{\gamma }))\implies S_{n}\rightarrow b(\gamma )\in S_{2}
\end{equation}
\end{definition}

\bigskip

\noindent This problem will be addressed in the next section.

\section{A general solution} \label{sec3}

Define the following statistics\textbf{\ } 
\begin{equation}
A_{n}(1,k,\ell )=k^{-1}\sum_{j=\ell +1}^{j=k}\ \sum_{i=j}^{i=k}i\delta
_{ij}\left( Y_{n-i+1,n}-Y_{n-j+1,n}\right) \left(
Y_{n-j+1,n}-Y_{n-j,n}\right) , \label{loe}
\end{equation}

\noindent where $i\delta _{ij}=\frac{1}{2}$ $if$ $i=j$ $\ $and $\delta _{ij}=1$ $if$ $%
\ i\neq $ $j$ $(k,\ell ),$ is a couple of integers such that $1\leq \ell
<k<n,$\textbf{\ } 
\begin{equation*}
y_{0}-Y_{n-k,n},1<k<n
\end{equation*}%
when $x_{0}(G)=y_{0}<+\infty .$ From these two statistics and from $%
T_{n}(2,k,\ell ),$ we form our ECSFEXT. Before we go any further, we should
remark that $A_{n}(1,k,\ell )$ was new in 1989. We discovered later that is
related to that of de Dekkers et al. \cite{deh} (1989)
\begin{equation}
M_{n}^{(2)}(k)=\frac{1}{k}\sum_{j=1}^{k}(Y_{n-k+1,1}-Y_{n-k,n})^{2}, \label{dekkers}
\end{equation}

\noindent in the sense that $A_{n}(1,k,1)=2M_{n}^{2}(k)$. We establish this in Subsection \ref{subsecapp2} of the Appendix Secion \ref{secapp}. We shall use this remark to
rediscover the result of de Dekkers \textit{et al.} \cite{deh} in new ways. Here is our ECSFEXT 
\begin{equation*}
S_{n}=^{t}(T_{n}(1),..,T_{n}(9))\in \mathbb{R}^{9},
\end{equation*}%

\noindent where  
\begin{equation*}
\ T_{n}(2)=T_{n}(2,k,\ell ),
\end{equation*}

\begin{equation*}
T_{n}(1,k,\ell )=T_{n}(2,k,\ell )A_{n}(1,k,\ell )^{-1/2},
\end{equation*}%

\begin{equation*}
\ T_{n}(5)=T_{n}(2,\ell ,1),
\end{equation*}

\begin{equation*}
T_{n}(6)=T_{n}(2,\ell ,1)/\left( Y_{n-\ell ,n}-Y_{n-k,n}\right) ,
\end{equation*}

\begin{equation*}
T_{n}(7)=A_{n}(1,\ell ,1)/\left( Y_{n-\ell ,n}-Y_{n-k,n}\right) ^{2},
\end{equation*}

\begin{equation*}
T_{n}(8,v)=n^{-v}\left( Y_{n-\ell ,n}-Y_{n-k,n}\right) ^{-1}
\end{equation*}

\noindent and, when $Y_{n,n}\uparrow y_{0}<+\infty ,$

\begin{equation*}
T_{n}(9)=(y_{0}-Y_{n-\ell ,n})/\left( y_{0}-Y_{n-k,n}\right) .
\end{equation*}%

\noindent Finally put

\begin{equation*}
\mathcal{S} =\text{ }\mathbb{R}_{+}^{2}\times \{0\}^{2}\times \overset{-}{\mathbb{R}}_{+}\times
\mathbb{R}_{+}\times \{0\}^{2}\subset \mathbb{\mathbb{R}}^{9}.
\end{equation*}

\noindent We denote by $\pi _{p,n},$ the projection of $\mathbb{R}^{n}$ onto $\mathbb{R}^{p}$ when $p<n$.
We begin to state the estimation of the extremal domain.

\bigskip

\begin{theorem} \label{theo22} Let $F\in \Gamma ,$ then for all sequences $k=k(n)=\left[ n^{\alpha
}\right] ,\ell =\left[ n^{\beta }\right] ,\frac{1}{2}<\beta <\alpha <1,$ $%
for $ any $v>0,$

\begin{itemize}
\item[(i)] $\pi _{4,9}(S_{n})$ converges almost surely to some $\pi _{4,9}(A),$ \ 
$A\in \pi _{4,9}(\mathcal{S})$, as $n\rightarrow +\infty $.

\item[(ii)] $\pi _{7,9}(S_{n})$ converges in probability to some  $A \in pi _{4,9}(\mathcal{S})$. Specifically,

\begin{itemize}
\item[(ii-a)] If $F\in D(\Lambda ),$ \ then $A=\left( 1,0,0,y_{0},0,0,0\right) $ $;$ 
$y_{0}=+\infty $ or \ $y_{o}<+\infty$.

\item[(ii-b)] If $F\in D(\phi _{\gamma }),$\ \ then $A=\left( 1,\gamma ^{-1},0,+%
\text{ }\infty \text{ ,}\gamma ^{-1},0,0\right) $ , for $\gamma>0$.

\item[(ii-c)] If $F\in D(\psi _{\gamma }),$ then $A=\left[ 1-(2+\gamma )^{-1}\right]
^{\frac{1}{2}}(0,0,y_{0},0,0,0)$ $;$ $\ for$ $y_{0}=+\infty $ ,$\gamma>0$.
\end{itemize}

\item[(iii)] In addition,
\end{itemize}
\end{theorem}

\textbf{4)} If $F\in D(\Lambda )\cup D(\Phi )$, then $T_{n}(8)\rightarrow 0,$ 
$a.s$., as $n\rightarrow +\infty .$

\textbf{5)} If $D(\psi _{\gamma })$, then $T_{n}(9)\rightarrow 0$, a.s., as $%
n\rightarrow +\infty$.\\

\begin{remark}
At this stage we see that the couple $(T_{n}(1), T_{n}(2))$, and then the couple $(A_{n}, T_{n}(2))$, suffices to estimate
the whole domain of attraction. One would like to have it as an ECSFEXT. Unfortunately, we need more other statistics to
achieve the full emprical characterization in Theorem \ref{theo23} below.  
\end{remark}

\begin{remark}
It would be legimitade to ask whether other simple statistics have this empirical characterization property for the Frechet domain of attraction as the Hill one does. At this stage, we simply remark that multiples and one-to-one functions of Hill's statistic surely inherit this property. It is the case for our statistic (\ref{loe}) and the de Dekkers \textit{et al.} (\ref{dekkers}) moment estimator. As for the de Haan and Resnick (\ref{haanresnick}) statistic, which is a estimator of the extremal index when this latter is positive, does not possess this property as we show it in Subsecion \ref{subsecapp1} of Section \ref{secapp}.
\end{remark}

\noindent By inverting the above theorem in the sense of the preceeding remarks, we get the ECSFEXT $S_{n}=(T_{n}(1)...,T_{n}(9))$. We have

\begin{theorem} \label{theo23} Let $k=\left[ n^{\alpha }\right] ,\ell =\left[ n^{\beta }\right] ,%
\frac{1}{2}<\beta <\alpha <1,$ \ \ \ $0<\delta <\frac{1}{2}$ such that $%
\delta +\beta -1>0,$ \ \ $2v=\min (1-\alpha ,\delta +\beta -1).$ Suppose
that $\pi _{7,9}(S_{n}) \rightarrow _{p}A \in \pi_{7,9}(\mathcal{S})$ as $n\rightarrow +\infty $ with $y_{0}=+\infty
.$ If further $T_{n}(8,\beta /2)\rightarrow _{p}0$ as $n\rightarrow +\infty$, then
\begin{itemize}
\item[(i)] $F\in D(\Lambda )$ whenever $c=1$ and $d=f=0$

\item[(ii)] $F\in D(\phi _{\gamma })$ whenever $c=1$ and $d=f=\gamma ^{-1},\gamma
>0.$
\end{itemize}

\noindent Now suppose that $\pi _{7,9}(S_{n}) \rightarrow _{p}A \in \pi_{7,9}(\mathcal{S})$ as $n\rightarrow +\infty $ with $%
y_{0}<+\infty$. Then

\begin{itemize}
\item[(iii)] If $c=1,$ \ $d=f=0$ and $T_{n}\left( 8,\beta /2\right) \rightarrow
_{p}0$ as $n\rightarrow +\infty ,$ then $F\epsilon D(\Lambda ).$

\item[(iv)] If $1\leq c<\sqrt{2},d=f=0$ and $T_{n}(9)\rightarrow _{p}0,$ then $%
F\in D(\psi _{\gamma }),$ where $\gamma =-2+c^{2}/\left( c^{2}-1\right) .$
\end{itemize}
\end{theorem}

\begin{remark}
We clearly get her an ECSFEXT of nine statistics. The only concern is that the number is relatively high, since we need only one statistic for the
Frechet domain. The main difficulty concerns de Gumbul subdomain. The idea behind the result of Mason is that the limit of the first asymptotic moment $R(x,G)$ to a positive number is equivalent to $F$ belongs to $D(\phi)$. For $F \in D(\psi_{\gamma})$, $R(x,G)$ goes to zero as $(x_{0}(G)-G^{-1}(1-k/n))/(\gamma+1)$. But for $F \in D(\Lambda)$, $R(x,G)$ has as many as possible ways to tend to zero. This explains why the characterization of $df$'s in $D(\Lambda)$ requires a considerable number of statistics.
\end{remark}

\begin{remark} \label{remdioplo}
Diop and Lo (1994) \cite{dioplo94} claimed an ECSFEXT of two statistics. They introduced the continuous generalized Hill's estimator
$$
S_{n}(\tau)= k^{-\tau} \sum_{i=1}^{k} i^{k} (\log X{n-i+1,n} - \log X_{n-i,n})
$$

\noindent where $\tau >0$ and $k$ satisfies the usual condition and further thoroughly studied it in \cite{dioplo} and \cite{dioplo90}. They indeed claimed that any couple of statistics $(S_{n}(\tau), S_{n}(\rho)$, for $\tau \neq \rho$, empirically characterizes the whole extremal domain of attraction. Further they acknowledged that their proof is wrong. However, any couple $(S_{n}(\tau), S_{n}(\rho)$, $\tau \neq \rho $, is indeed an ESFEXT.
\end{remark}

\section{Proofs of the theorems} \label{sec4}

\noindent Introduce the two first asymptotic moments 
\begin{equation}
R(x,z,F)=(1-F(x))^{-1}\int_{x}^{z}(1-F(t))dt, \label{fm1}
\end{equation}%
\ \ $x<z\leq x_{0}(F)$ \ with $R(x,x_{0},F)\equiv R(x,F)$ and

\begin{equation}
W(x,z,F)=(1-F(x))^{-1}\int_{x}^{z}\int_{y}^{z}(1-F(t))dtdy, \label{fm2}
\end{equation}

\noindent $x<z\leq x_{0}(F)$ with $W(x,x_{0},F)\equiv W(x,F).$ Let 
\begin{equation*}
F^{-1}(u)=\inf \left\{ x,\text{ \ }F(x)\geq u\right\} ,
\end{equation*}

\noindent $0<u\leq 1,\ F^{-1}(0)=F^{-1}(0_{+})$, the generalized inverse function of $F$ and let also $F(1)=1$.\\

\noindent From now on, $R(x,\cdot )$ and $W(x,\cdot )$ are used only for $%
G(x)=F(e^{x}).$ The proofs are based on the technical
tools in Section \ref{sectools}. We first say that \textbf{Fact 1} in Section \ref{sectools} means

\begin{equation}
\left\{ Y_{j,n},1\leq j\leq n,\text{ \ \ }n\geq 1\right\} =_{d}\left\{
G^{-1}(U_{n-j+1}),\text{ \ }1\leq j\leq n,\text{ \ }n\geq 1\right\}
\label{f3.0}
\end{equation}

\noindent where $=_{d}$ stands for equality in distribution.\\

\noindent \textbf{Proof of Theorem \ref{theo22}}.\\

\noindent Let $F\in \Gamma ,$ \ \ \ \ $u_{n}=1-G(Y_{n-k,n}),$ \ $%
v_{n}=1-G(Y_{n-\ell ,n}),$ \ $k=\left[ n^{\alpha }\right] ,\ell =\left[
n^{\beta }\right] ,\frac{1}{2}<\beta <\alpha <1.$ First, we have to prove
that

\begin{equation}
0<v_{n}<u_{n}\rightarrow 0,a.s.,\text{ }and\text{ \ }v_{n}/u_{n}\rightarrow
0,  \label{f3.1}
\end{equation}

\bigskip

\noindent a.s. as $n\rightarrow +\infty$. By \textbf{Facts 1, 2 and 3} in Section \ref{sectools}, we have

$$
v_{n}/u_{n}=\frac{1-G_{n}(Y_{n-\ell ,n})+0(n^{-\delta })}{%
1-G_{n}(Y_{n-k,n})+0(n^{-\delta })}=\frac{U_{\ell +1,n}+0(n^{-\delta })}{%
U_{k+1,n}+0(n^{-\delta })}=
$$

\begin{equation}
=(\ell /k)\frac{1+0(n^{-\beta -\delta +1}}{%
1+0(n^{-\alpha -\delta +1})}\rightarrow 0,a.s.  \label{f3.2}
\end{equation}

\noindent as $n\rightarrow +\infty $ whenever $0<\delta <\frac{1}{2},$ \ \ \
\ $\beta +\delta -1>0$.\newline
\newline
The proof of Theorem \ref{theo22} will follow from the partial proofs of
Statements $(S1),(S2)$, etc.\newline

\noindent \textbf{(S1)} : $T_{n}(2,k,\ell )=R(x_{n})(1+o(1))$, a.s., as $%
n\geq 1$,\\

\noindent where $x_{n}=Y_{n-k,n}$ and $z_{n}=Y_{n-\ell ,n}$, $n\geq 1$.\\

\noindent \textbf{Proof of (S1)}. It is easy to check that

\begin{equation}
T_{n}(2,k,\ell )=nk^{-1}\int_{x_{n}}^{z_{n}}(1-G_{n}(t))\text{ }dt.
\label{f3.3a}
\end{equation}

\noindent By \textbf{Fact 2} in Section \ref{sectools}, we have for all $\delta ,$ $0<\delta <%
\frac{1}{2},$

\begin{equation}
T_{n}(2,k,\ell )=R(x_{n},z_{n})\left( 1+0(n^{-\alpha -\delta +1}(z_{n}-x_{n})%
\text{ }/\text{ \ }R(x_{n,z_{n}})\right) ,a.s.,  \label{f3.3b}
\end{equation}

\noindent as $n\rightarrow +\infty$. We choose $\delta $ so that $\mu =\beta +\delta -1>0.$ \ By Lemmas %
\ref{lemmao7} and \ref{lemmao8} of in Section \ref{sectools}, Statements (\ref{f3.0})
and (\ref{f3.1}), we have

\begin{equation}
T(x_{n},z_{n})/R(x_{n})\rightarrow 1,\text{ \ }a.s.,\text{ \ \ \ }as\text{ \
\ }n\rightarrow +\infty  \label{f3.4}
\end{equation}

\noindent and then

\begin{equation}
T_{n}(2,k,\ell )=R(x_{n})(1+0\left( n^{-\mu}(z_{n}-x_{n}\right) \text{ \ }/%
\text{ \ }R(x_{n})),a.s.,\text{ }as\text{ }n\rightarrow +\infty.  \label{f3.5}
\end{equation}

\bigskip 

\noindent Now \textbf{either }$ G\in D(\phi _{\gamma }),$ \ \ $\gamma >0$ and we
apply Formula 2.6.4 of De Haan (1970) (\cite{dehaan}):

$$
R(x_{n})/(y_{0}-x_{n})\rightarrow q=(1-K)K^{-1},a.s.,
$$

\bigskip 
\noindent as $n\rightarrow +\infty$ \ with $K=(1-1/(\gamma +2))$ and $\frac{1}{2}<1,$ to
get

\begin{equation}
0\leq n^{-\mu}(x_{n},z_{n})/R(x_{n}) \leq
\left\{2c^{-1}(x_{n},z_{n})/(y_{0}-x_{n})\right\} n^{-\mu} \leq \frac{2}{q}%
n^{-\mu}\rightarrow 0,\text{ },
\label{f3.6}
\end{equation}

\noindent a.s., as $n\rightarrow +\infty$.\\

\noindent \textbf{Or} $F\in D(\Lambda \ )\cup D(\phi _{\gamma }\ )$ for
some $\gamma >0.$ By Lemma \ref{lemmao4} in Section \ref{sectools}, $u^{1/\gamma
}F^{-1}(1-u)$ is SVZ when $F\in D(\phi _{\gamma }\ ).$ It may also be
shown from (\ref{f23a}) that $F^{-1}(1-u)$ is SVZ when $F\in D(\Lambda)$. Now, using (\ref{f3.0}) and the Karamata representation for $F^{-1}(u)$, we get in both cases for $0<\varepsilon <\gamma $,

\begin{equation}
\frac{1}{2}(u_{n}/v_{n})^{\lambda -\varepsilon } \leq F^{-1}(X_{n-\ell
,n})/F^{-1}(X_{n-k,n})=F^{-1}(X_{n-k,n})  \label{f3.7a}
\end{equation}

$$
=F^{-1}(1-U_{\ell +1,n})/F^{-1}(1-U_{k+1,n}) \leq 2(u_{n}/v_{n})\lambda
+\varepsilon,
$$

\noindent a.s., as $n\rightarrow +\infty$ , where $\lambda =1/\gamma$ when $%
F\in D(\phi _{\gamma })$ or $\lambda =0$ when $F\in D(\ \Lambda \ )$. Since $%
G^{-1}(1-u)=\log F^{-1}(1-u)$ and since $u_{n}/v_{n} \sim n^{\alpha -\beta}$%
, a.s. as $n\rightarrow +\infty$, it follows that

\begin{equation}
z_{n}-x_{n}=0(\log n), \text{ } a.s. \text{ } n\rightarrow +\infty,
\label{f3.7b}
\end{equation}

\noindent and thus

\begin{equation}
0\leq \left( z_{n}-x_{n}\right) \ /\ \left( n^{\mu}R(x_{n})\right) =0\left( 
\frac{\log n}{n^{\mu/2}}\times \frac{1}{n^{\mu/2}R(x_{n})}\right), \text{ }
a.s. \text{ } as \text{ }n\rightarrow +\infty  \label{f3.8}
\end{equation}

\noindent By Lemma \ref{lemmao6} in Section \ref{sectools}, $R\left(
G^{-1}(1-u)\right) $ is SVZ, and since $u_{n}\ \sim \ (k/n),\ a.s.\
n\rightarrow +\infty $, $R(x_{n})\sim R(G^{-1}(1-k/n))$, a.s., as $%
n\rightarrow +\infty $. Hence, by Lemma \ref{lemmao4} in Section \ref{sectools},

\begin{equation}
n^{\mu/2}R(x_{n}) \sim n^{\mu/2}R(G^{-1}(1-k/n))\rightarrow +\infty ,\qquad
a.s.,\ as\qquad n\rightarrow +\infty .\qquad  \label{f3.9}
\end{equation}

\bigskip

\noindent (\ref{f3.3b}), (\ref{f3.5}), (\ref{f3.6}), (\ref{f3.8}) and (\ref%
{f3.9}) together prove \textbf{(S1)}.\newline
\newline

\noindent \textbf{(S2)} : 
 
\begin{equation*}
A_{n}(1,k,\ell) =\left( T_{n}(2,k,\ell )/T_{n}(1,k,\ell)\right)^{2}
=W(x_{n})\left( 1+o(1)\right)=K \times R(x_{n})^{2}(1+o(1)),
\end{equation*}

\noindent a.s., as $n\rightarrow +\infty$, where $K=1$ if $F\in D(\Lambda) \cup
D(\phi) $, and $K=1-1/(\gamma+2)$ if $F\in D(\psi _{\gamma }),\gamma>0$.%
\newline
\newline
\textbf{Proof of (S2)}. We check that

\begin{equation*}
A_{n}^(1,k,\ell )=nk^{-1}\int_{x_{n}}^{z_{n}}\ \int_{y}^{z_{n}}1-G_{n}(t)dt\
dy.
\end{equation*}

\noindent By \textbf{Fact 2 } in Section \ref{sectools}, 

\begin{equation}
A_{n}(1,k,\ell )=W(x_{n},z_{n})(1+0(n^{-\mu
}(z_{n}-x_{n})^{2}W(x_{n},z_{n})),
\label{f3.11}
\end{equation}

\noindent a.s.,as $n\rightarrow +\infty$. By Lemma \ref{lemmao7} in Section \ref{sectools}, and by Statements (\ref%
{f3.0}) and (\ref{f3.1}), 

\begin{equation}
W(x_{n})/R(x_{n})\rightarrow 1, 
\label{f3.12}
\end{equation}

\noindent a.s. as $n\rightarrow +\infty$. Hence, Lemma \ref{lemmao2} in Section \ref{sectools}$\ $yields

\begin{equation}
W(x_{n})/R(x_{n})^{2}\rightarrow K, \ a.s.,as\ \ n\rightarrow +\infty \ 
\label{f3.13}
\end{equation}

\noindent It follows from (\ref{f3.11}), (\ref{f3.12}) and (\ref{f3.13}) that

\begin{equation}
A_{n}(1,k,\ell )= W(x_{n})(1+0(n^{-\mu})(z_{n}-x_{n})^{2}R(x_{n})^{-2})),\
a.s.,as\ \ n\rightarrow +\infty  \label{f3.14}
\end{equation}

\noindent But the calculations that led to (\ref{f3.6}) and (\ref{f3.8})
showed that for all $\rho >0$, $0<\zeta \leq \xi$,

\begin{equation}
\ (x_{n},z_{n})^{\zeta_{n}-\rho}R(x_{n})^{-\xi} \rightarrow 0\ a.s.,as\ \
n\rightarrow +\infty  \label{f3.15}
\end{equation}

\noindent whenever $F\in \Gamma$. Thus (\ref{f3.13}) and (\ref{f3.14})
ensure $(S2)$.\newline
\newline
\noindent \textbf{(S3)} : $T_{n}^{{}}(1,k,\ell )\rightarrow K^{-1}K\sqrt{K}$%
, a.s., as $n\rightarrow +\infty$.\newline
\newline
\textbf{Proof of (S3)}. (S1) and (S2) prove (S3).\newline
\newline
\textbf{(S4)} : $T_{n}^{{}}(3,k,\ell ,v)\rightarrow 0$, a.s., as $%
n\rightarrow +\infty$.\newline
\newline
\noindent \textbf{Proof of (S4)}. (S1) implies that $T_{n}^{{}}(3,k,\ell ,v)\sim
n^{-v}(x_{n},z_{n})R(x_{n})^{-1}$, a.s., as $n\rightarrow +\infty$. Thus (%
\ref{f3.15}) completes the proof of (S4).\\

\noindent \textbf{(S5)} : $T_{n}(4)\uparrow y_{o}$, a.s., as $n\rightarrow
+\infty $.\newline
\newline
\noindent \textbf{Proof of (S5)}. This fact is obvious.\newline
\newline
\noindent \textbf{(S6)}. We have%
\begin{equation*}
T_{n}^{{}}(2,\ell ,1)\rightarrow \left\{ 
\begin{array}{c}
1/\gamma ~\qquad if\quad F\in D(\phi _{\gamma }) \\ 
0~\qquad if\quad F\in D(\psi ) \\ 
0~\qquad if\quad F\in D(\Lambda).%
\end{array}%
\right.
\end{equation*}%
\bigskip

\noindent \textbf{Proof of (S6)}. If $F \in D(\phi _{\gamma }), \text{ }
T_{n}^{{}}(2,\ell ,1)\rightarrow 1/\gamma , \text{ } a.s. \text{ } as \:
n\rightarrow +\infty$ by Theorem 2 of Mason (1982) (\cite{mason}). For $%
F \in D(\Lambda )\cup D(\psi )$, use (\ref{f3.3a}) and get, for $0<\mu <\frac{1}{2}$,

\begin{equation}
T_{n}(2,\ell ,1)\leq R(z_{n})+0\left( n^{-\mu }(Y_{n,n}-Y_{n-\ell,n})\right)
\end{equation}

\begin{equation*}
\leq :R(z_{n})+0(n^{-\mu }\alpha _{n}),
\end{equation*}

\noindent  a.s. as $n\rightarrow +\infty$. But $\alpha _{n}\rightarrow 0$ as $n\rightarrow +\infty$ when $%
F \in D(\psi )$ since $Y_{n,n} \uparrow y_{0}$ and $Y_{n-\ell,n} \uparrow y_{0}<+\infty$. If $F \in D(\Lambda)$, it may be showed as in (\ref{f3.7b}) that 
$\alpha _{n}=0_{p}(\log n)$ as $n\rightarrow +\infty$,
that is

\begin{equation}
\lim_{\rho \uparrow +\infty} \mathbb{P}\left( \alpha_{n}>\rho \log n\right)=0.
\end{equation}

\noindent Hence in both cases, $T_{n}(2,\ell ,1)\rightarrow_{\mathbb{P}}0,$
since $R(x_{n})\rightarrow_{\mathbb{P}}0$, as  $n\rightarrow +\infty ,$
by Lemma \ref{lemmao1}. The proof of $(S6)$ is now complete.\\

\noindent \textbf{(S7)}. $T_{n}(6)\underset{p}{\rightarrow }0$ as $n\rightarrow +\infty $.\\

\noindent \textbf{Proof of (S7)}. We use the device of Fact 5 in (\ref{f3.3a})
by considering the integral as an improper one with respect to the upper
bound. Remarking that $\left( \ell /n\right) \quad \sim \quad \left(
1-G(Y_{n-\ell ,n})\right)$, we get

\begin{equation}
T_{n}(6)\leq 2 Z_{n}(1)R(z_{n})/(z_{n}-x_{n})  \label{f3.17}
\end{equation}

\bigskip

\noindent where $Z_{n}(1)=\sup_{U_{1,n\leq s\leq 1}}\left\vert
U_{n}(s)/s\right\vert $. This together with Fact 4 and Lemma \ref{lemmao8}
in Section \ref{sectools} ensures $(S7)$.\newline
\newline
\noindent \textbf{(S8)} $T_{n}(7)\underset{p}{\rightarrow }0$, as $%
n\rightarrow +\infty $,\newline
\newline
\noindent \textbf{Proof of (S8)}. As for $T_{n}(6)$, we have

\begin{equation}
T_{n}(7)\leq 2Z_{n}(1)W(z_{n})/(z_{n}-x_{n})^{2}\underset{p}{\rightarrow }0 
\text{ } as \text{ } n\rightarrow +\infty ,  \label{f3.18}
\end{equation}

\noindent by the very same arguments. Thus $T_{n}(7)\rightarrow_{\mathbb{P}}%
0$, as $n\rightarrow +\infty$, is proved.\newline
\newline
\noindent \textbf{(S9)}. If $F\in D(\Lambda )\cup D(\phi
),T_{n}(8)\rightarrow 0$, a.s., as $n\rightarrow +\infty $.\newline
\newline
\noindent \textbf{Proof of (S9)}. We recall that $%
T_{n}(8)=n^{-v}(z_{n}-x_{n})^{-1}$. By the DDHM's representation (cf. Lemma %
\ref{lemmao4} in Section \ref{sectools} and by (\ref{f3.0})), we have for all $\lambda
>1$,

\begin{equation}
z_{n}-x_{n}\geq G^{-1}\left( 1-\lambda \varepsilon _{n}\right) =s\left(
\varepsilon _{n}\right) +\int_{\varepsilon _{n}}^{\lambda \varepsilon
_{n}}s(t)t^{-1}dt,  \label{f3.19}
\end{equation}

\noindent for large $n$, where $\varepsilon _{n}=U_{\ell +1,n}.$ Now, the
properties of SVZ functions easily yield for any fixed $\varepsilon ,\quad
0<\varepsilon <\frac{1}{2},\quad -\varepsilon +\left( 1-\varepsilon \right)
\log \lambda \geq 1,$

\begin{equation}
z_{n}-x_{n}\geq \left( 1-\varepsilon \right) s\left( \ell /n\right) \left\{
-\varepsilon +\left( 1-\varepsilon \right) \log \lambda \right\} \geq \frac{1%
}{2}s\left( \ell /n\right) ,a.s.,
\label{f3.20}
\end{equation}

\noindent as $n\rightarrow +\infty$. Thus, by Lemma \ref{lemmao4} in Section \ref{sectools}, 
\begin{equation}
n^{\nu }(x-y)\geq \frac{1}{2}ns(\ell /n)\rightarrow +\infty ,
\end{equation}%
\noindent a.s. as $n\rightarrow +\infty $. The proof of $(S10à$ is complete.\newline
\newline
\noindent \textbf{(S10)}. If $F\in D(\psi )$, then $T_{n}(9)\rightarrow 0$,
a.s., as $n\rightarrow +\infty $.\\

\noindent \textbf{Proof of (S10)}. It is already obtained in (\ref{f3.20}).\\

\noindent We now sum up our partial proofs to get Theorem \ref{theo22} :

\begin{itemize}
\item[(i)] (S3) gives the two possible limits of $T_{n}\left( 3,k,\ell,v\right) $

\item[(ii)] (S1) and Lemma \ref{lemmao1} in Section \ref{sectools} \ give the two
possible limits of $T_{n}\left( 2,k,\ell \right) .$

\item[(iii)] (S4) gives the unique limit of $T_{n}(3,k,\ell ,v).$

\item[(iv)] (S5) gives the limits of $T_{n}(4)$

\item[(v)] (S6) gives the two possible limits of $T_{n}(2,\ell ,1).$

\item[(vi)] (S7) gives the limit of $T_{n}(6)$

\item[(vii)] (S8) gives the limit of $T_{n}(7).$
\end{itemize}

\bigskip

\noindent These points ensure Parts (i) and (ii) of Theorem \ref{theo22}. As to the part (iii),
it is proved by (S9) and (S10).

\subsection{Proof of Theorem \ref{theo23}}

First, use \textbf{Fact 2} in Section \ref{sectools} \ as in (\ref{f3.3a}) and get

\begin{equation}
R(x_{n},z_{n})=T_{n}(2,k,\ell )=T_{n}(2,k,\ell
)(1+0_{p}(n^{-2v}(z_{n}-x_{n})\ /\ T_{n}(2,k,\ell )).
\end{equation}

\begin{equation}
=T_{n}(2,k,\ell )(1+0_{p}(T_{n}(3,k,\ell ,2v))
\end{equation}

\noindent But $T_{n}(3,k,\ell ,2v)\underset{p}{\rightarrow }0$ and then as $n\rightarrow +\infty$.

\begin{equation}
R(x_{n},z_{n})=T_{n}(2,k,\ell )(1+0_{p}(1)),
\label{f3.22}
\end{equation}

\noindent Secondly,

\begin{equation}
W(x_{n},z_{n})=A_{n}(1,k,\ell )(1+0_{p}(n^{-2v}(z_{n}-x_{n})\ /\
A_{n}(1,k,\ell )).  \label{f3.23}
\end{equation}

\bigskip 
\noindent as $n\rightarrow +\infty$. Since $T_{n}(1,k,\ell )\underset{p}{\rightarrow }%
K^{-1},c^{2}=K^{-1},c^{2}=K^{-1},A_{n}(1,k,\ell)=K^{2}+0_{p}(1)T_{n}(2,k,\ell )^{2}$ as $n\rightarrow +\infty$. Thus

\begin{equation}
W(x_{n},z_{n})=K\ \ T_{n}(2,k,\ell )^{2}(1+0_{p}(T_{n}(3,k,\ell
,v)^{2}))),as\ \ n\rightarrow +\infty  \label{f3.24a}
\end{equation}

\bigskip

\noindent Since $T_{n}(3,k,\ell ,v)\underset{p}{\rightarrow }0,as\ \
n\rightarrow +\infty$, one has

\begin{equation}
W(x_{n},z_{n})=K\ \ T_{n}(2,k,\ell )^{2}(1+0_{p}(1)),as\ \ n\rightarrow
+\infty.  \label{f3.24b}
\end{equation}

\bigskip

\noindent Formulas (\ref{f3.22}) and (\ref{f3.24b}) together imply

\begin{equation}
W(x_{n},z_{n})\ /\ R(x_{n},z_{n})^{2}\rightarrow _{p}K,\ as\ \ n\rightarrow
+\infty.  \label{f3.24c}
\end{equation}

\noindent We now want to drop $z_{n}$ in (\ref{f3.24c}). It suffices to check whether
the conditions of Lemma \ref{lemmao7} are satisfied. For that, we use the
device of Fact 5 for $T_{n}(2,k,\ell )$ see (\ref{f3.17}) to get

\begin{equation}
T_{n}(2,k,\ell )\geq Z_{n}(2)(R(z_{n})-(n/\ell )(1-G(Y_{n,n}))R(Y_{n,n})),
\label{f3.25a}
\end{equation}

\bigskip

\noindent where $Z_{n}(2)=\inf_{U_{1,n}\leq s\leq 1}\left| U_{n}(s)\ /\
s\right|$. But for all $d.f.G$  $G$, $(G^{-1}(u))\geq u$. Thus, by
applying (\ref{f3.0}), one has

\begin{equation}
-\frac{n}{\ell }(1-G(G^{-1}(1-U_{1,n}))\geq nU_{1,n}\ /\ \ell
=-0_{p}^{+}(1)\ /\ \ell ,as\ \ n\rightarrow +\infty.  \label{f3.25b}
\end{equation}

\noindent where $X=0_{p}^{+}(1)$ means that $P(X<0)=0$ and $X=0_{p}(1)$. By
Statement (\ref{f3.39b}) below, $R(Y_{n,n})=0_{p}^{+}(1)\ as\ \ n\rightarrow
+\infty .$ Finally, we arrive at

\begin{equation}
T_{n}(6)\geq Z_{n}(2)\left( R(z_{n})\ /\ (z_{n}-x_{n}\right) -n^{-\beta
}(z_{n}-x_{n})^{-1}0_{p}(1)),\ as\ \ n\rightarrow +\infty  \label{3.26}
\end{equation}

\bigskip

\noindent and

\begin{equation}
T_{n}(6)\geq Z_{n}(2)\left( R(z_{n})\ /\ (z_{n}-x_{n}\right) -T_{n}(8,\beta
)0_{p}(1)), as n\rightarrow +\infty  \label{f3.27a}
\end{equation}

\bigskip

\noindent But $T_{n}(8,\beta /2)\ \underset{p}{\rightarrow }0$ by
assumption. Thus $T_{n}(8,\beta )\underset{p}{\rightarrow }0$ and by (\ref%
{f3.27a}) an by Fact 4,

\begin{equation}
R(z_{n})\ /\ (z_{n}-x_{n})\underset{p}{\rightarrow }0, as n\rightarrow
+\infty  \label{f3.27b}
\end{equation}

\noindent We also have, if $y_{0}<+\infty ,$

\begin{equation}
R(z_{n})\ /\ (z_{n}-x_{n})\leq (y_{0}-z_{n})\ /\
(z_{n}-x_{n})=(-1+T_{n}(9)^{-1})^{-1}  \label{f3.27c}
\end{equation}

\bigskip

\noindent If $T_{n}(9)\underset{p}{\rightarrow }0$, i.e., $%
T_{n}(9)^{-1}\rightarrow +\infty ,$ then

\begin{equation}
R(z_{n})\ /\ (z_{n}-x_{n})\underset{p}{\rightarrow }0,as\ \ n\rightarrow
+\infty.  \label{f3.27d}
\end{equation}

\noindent This is the first condition of Lemma \ref{lemmao7}. For the second,
we remark that

\begin{equation}
A_{n}(1,k,\ell )\geq Z_{n}(2)\left( W(z_{n})-\ell
^{-1}0_{p}(1)W(Y_{n,n})\right) ,as\ \ n\rightarrow +\infty  \label{f3.28a}
\end{equation}

\noindent Since $W(Y_{n,n})=0_{p}^{+}(1)$ by Lemma \ref{lemmao1} and
Statement (\ref{f3.43b}) below, we obtain

\begin{equation}
T_{n}(7)\geq Z_{n}(2)\left\{ W(z_{n})\ /\ (zn-x_{n})^{2}-T_{n}(8,\beta
)^{2}0_{p}^{+}(1)\right\}.  \label{f3.28b}
\end{equation}

\bigskip

\noindent By the same reasons that gave (\ref{f3.27b}), we arrive at

\begin{equation}
W(z_{n})\ /\ (zn-x_{n})\underset{p}{\rightarrow }0,\text{ }as\ \
n\rightarrow +\infty,  \label{f3.28c}
\end{equation}

\noindent whenever $T_{n}(8,\beta )\underset{p}{\rightarrow }0,$ $as\ \
n\rightarrow +\infty .$ We also have, when $y_{0}<+\infty ,$

\begin{equation}
W(z_{n})\ /\ (zn-x_{n})^{2}\leq (y_{0}-z_{n})^{2}/\
(zn-x_{n})^{2}=(-1+T_{n}(9)^{-1})^{-2}.  \label{f3.28d}
\end{equation}

\bigskip 
\noindent Hence $T_{n}(9)\underset{p}{\rightarrow }0$, as $n\rightarrow+\infty$, implies

\begin{equation}
W(z_{n})\ /\ (zn-x_{n})\underset{p}{\rightarrow }0,\text{ }as\ \
n\rightarrow +\infty.  \label{f3.28e}
\end{equation}

\noindent We have proved that the conditions of Lemma \ref{lemmao7} are
satisfied via (\ref{f3.27b}), (\ref{f3.27b}), (\ref{f3.27d}), (\ref{f3.28c})
and (\ref{f3.28e}). Thus (\ref{f3.24c}) becomes

\begin{equation}
W(x_{n})/R(x_{n})^{2}\rightarrow _{p}K,as\ \ n\rightarrow +\infty ,\ with\
1\leq K<\sqrt{2}.  \label{f3.29}
\end{equation}

\noindent We now show how the preceeding may prove Theorem \ref{theo22}

\noindent If $T_{n}\rightarrow _{p}A\ with\ d=1/\gamma ,\quad 0<\gamma
<+\infty ,\quad $thus by Mason (1982) (\cite{mason}), $F\,\in D(\phi
_{\gamma })$ and all the other limits of $T_{n}(i),i=1,...,8$ are justified
by Theorem \ref{theo21}.\newline
\newline

\noindent If $T_{n}\rightarrow _{p}A\ with\ c=1$ and \ $T_{n}(8,\beta
)\rightarrow _{p}0,$ thus we get (\ref{f3.29}) with $K=c^{-2}=1$.\newline
\newline
\noindent If $T_{n}\rightarrow _{p}A\ with\,\ y_{0}<+\infty $ and $1\leq c<%
\sqrt{2}$ and $T_{n}(+9)\rightarrow _{p}0,$ thus (\ref{f3.29}) holds with $%
K=c^{-2}=1-1/(\gamma +2),$ for some $\gamma$, $0<\gamma <+\infty$.\\

\noindent It is clear now that Lemma \ref{lemmao2} ensures from (ii)
that $F\in D(\Lambda )$ and from (iii) that $F\in D(\psi_{\gamma })$ if we
show

\begin{equation}
\lim_{x\rightarrow y_{0}^{{}}}W(x)/R(x)^{2}=\lim_{n\rightarrow +\infty
}W(x_{n})/R(x_{n})^{2}.  \label{f3.30}
\end{equation}

\noindent \textbf{Proof of (\ref{f3.30})}. Recall basic facts

\begin{equation}
n^{-\rho }\leq n^{-\rho }+n^{-1}(k+1)/n\leq ^{-\rho }+2n\quad for\quad k= 
\left[ n^{\alpha }\right] ,\quad \rho =1-\alpha
\end{equation}

\noindent By Fact 2, for $0<\delta <\frac{1}{2},0<\delta <+\alpha -1<1,
0<\varepsilon <0.01$,

\begin{equation}
n^{-\rho }-\varepsilon n^{-\delta } \leq U_{k+1,n}\leq n^{-\rho
}+2\varepsilon n^{-\delta },
\end{equation}
a.s., as $n\rightarrow +\infty$.\\

\noindent Set $a_{n}=n^{-\rho },\quad n\geq 1$. This
sequence $a_{1}=1 > ... >a_{j}>aj+1>...$ makes a partition of $\left[ 0,1%
\right] .$ For $x\uparrow y_{0},u=1-G(x)\downarrow 0,$ there exists at each
step of this limit an integer $n$ such that $a_{n+2}\leq u\leq
a_{n+1}<a_{n}. $ Let $m=m(n)=n-\left[ n^{\tau }\right] ,\quad \tau =2-\alpha
-\delta .$ Remark that $\frac{1}{2}<\tau <1.$ One has $a_{m}-a_{n}=-\left[
n^{\tau }\right] f^{\prime }(\zeta _{n})$ with $f(x)=x^{-\rho }$ and $m\leq
\zeta _{n}\leq n$. Thus for large values of $n,$

\begin{equation}
a_{m}-a_{n}\geq (1-\varepsilon )n^{-\delta }.
\end{equation}

\noindent By the preceeding facts, for large $n,$

\begin{equation}
U_{k(m)+1,m}\geq a_{m}-\varepsilon (1+\varepsilon )n^{-\delta }.
\end{equation}

\noindent Since $m \rightarrow +\infty$ as $n\rightarrow +\infty$, it follows that

\begin{equation}
a_{n+2}\leq u\leq a_{n+1}<a_{n}\leq a_{m}-(1-\varepsilon )n^{-\delta }\leq
a_{m}-\varepsilon (1+\varepsilon )n^{-\delta }\leq U_{k(m)+1,m}
\label{f3.31}
\end{equation}

\noindent and

\begin{equation}
u/U_{k(m)+1,m}\rightarrow _{p}as\ \ n\rightarrow +\infty ,u\rightarrow 0
\label{f3.32}
\end{equation}

\noindent Put

\begin{equation*}
M(x)=\int_{x}^{y_{0}}1-G(t)dt \: \: and \: \:
m(x)=\int_{x}^{y_{0}}\int_{y}^{y_{0}}\: (1-G(t))\: dy\: dt.
\end{equation*}

\noindent Using the inequality $G^{-1}\left( G((x)\right) \leq x$ for all $x$
and for all $df$ $G$ and noticing that both $M(.)$ and $m(.)$ are nonincreasing, we obtain

\begin{equation}
0\leq M(G^{-1}(1-U_{k+1,m}))-M(G^{-1}(1-u))\leq
M(x_{m})-M(x)=\int_{x_{m}}^{x}1-G(t)dt.  \label{f3.33}
\end{equation}

\noindent Because of Lemma \ref{lemmao5}, either $x=G^{-1}(1-u)$ or $x$
lies on the constancy interval of $G,\left] G^{-1}(1-u),G^{-1}(1-u+)\right] =%
\left] y,z\right] .$ Hence

$$
0\leq M(x_{m})-M(x)\leq \int_{x_{m}}^{y}1-G(t)dt+(z-y)u
$$

\begin{equation}
\leq \left\{
G^{-1}(1-u)-x_{m}\right\} +(z-y)u.  \label{f3.34}
\end{equation}

\noindent One may quickly check that for large values of $n,$

\begin{equation}
\left\{ 
\begin{array}{c}
0\leq 1-G(x_{m}) \sim U_{k(m)+1,m} \sim u, \text{ ( see Statement (\ref{f3.1}%
))} \\ 
0\leq G^{-1}(1-u)-x_{m}\leq z_{m}-x_{m}, \\ 
0\leq z-y=G^{-1}(1-u+)-G^{-1}(1-u)\leq z_{m}-x_{m}%
\end{array}
\right.
\end{equation}

\noindent Hence,

$$
0\leq 1-(1+0_{p}(1))R(x)/R(x_{m})
$$

\begin{equation}
\leq (z_{m}-x_{m})/(m^{\alpha-1}R(x_{m}))+2(z_{m}-x_{m})/R(x_{m}),  \label{f3.35}
\end{equation}

\noindent as $n\rightarrow +\infty$. Using (\ref{f2.10}), we finally get

$$
0 \leq 1-(1+0_{p}(1))R(x)/R(x_{m})
$$

\begin{equation}
\leq (z_{m}-x_{m})/ (m^{\alpha
-1}R(x_{m}))+2\left( 1+\frac{z_{m}-x_{m}}{R(z_{m})}\right) ^{-1}\frac{%
z_{m}-x_{m}}{m^{\alpha -\beta }R(z_{m})}  \label{f3.36}
\end{equation}

\begin{equation}
\leq \frac{z_{m}-x_{m}}{m^{2v}R(z_{m})} +m^{\beta -\alpha }\underset{x\geq 0}%
{\sup }\left| x(1+x)\right| \leq \frac{z_{m}-x_{m}}{m^{2v}R(z_{m})}+m^{\beta
-\alpha},
\end{equation}

\noindent when $n$ is large enough. Now from (\ref{f3.3a})

\begin{equation}
m^{2v}T_{n}(2,k(m),\ell (m))/(z_{m}-x_{m})=T_{m}(3,k,\ell ,v)^{-1}\leq
Z_{m}(1)\frac{m^{2v}R(x_{m}-z_{m})}{z_{m}-x_{m}},  \label{f3.37}
\end{equation}

\noindent where $Z_{m}(1)$ is defined in (\ref{f3.17}). Since $T_{m}(3,v)%
\underset{p}{\rightarrow 0}$ as $\rightarrow +\infty$, by assumption, and
since $Z_{m}=0_{p}(1)$, we get

\begin{equation}
(z_{m}-x_{m})/(m^{2v}R(x_{m},z_{m}))\underset{p}{\rightarrow 0},\quad
as\quad n\rightarrow +\infty.  \label{f3.38}
\end{equation}

\noindent This and (\ref{f3.37}) together imply

\begin{equation}
\lim_{x\rightarrow y_{0}}R(x)/R(x_{m})=1\quad in\quad probability.
\label{f3.39a}
\end{equation}

\noindent It follows that

\begin{equation}
\lim_{x\rightarrow y_{0}}R(x) = \lim_{n\rightarrow +\infty}R(x_{n}) \quad
in\quad probability  \label{f3.39b}
\end{equation}

\noindent By the very same arguments, one gets

$$
0\leq 1-(1+0_{p}(1))W(x)/Ww(x_{m})
$$

$$
\leq m^{-(1-\alpha )}\left\{
G^{-1}(1-u)-x_{m}\right\} ^{2}/W(x_{m})+3(y-z)^{2}/W(x_{m})
$$

\begin{equation*}
\leq (z_{m}-x_{m})^{2}/(m^{2v}W(x_{n}))+3\frac{(x_{m}-z_{m})^{2}}{W(z_{m})}%
\frac{W(z_{m})}{W(x_{m})},\:  as \: n\rightarrow +\infty.
\end{equation*}

\noindent Now, by using (\ref{f2.15}), we arrive at

\begin{equation}
0\leq 1-(1+0_{p}(1))W(x)~/~W(x_{m})\leq
(z_{m}-x_{m})^{2}~/~(m^{2v}W(x_{m}))+3m^{(\alpha -\beta )}\underset{x\geq 0}{%
\sup }\left| x(1+x)\right| ,  \label{f3.41}
\end{equation}

\noindent as $n\rightarrow +\infty$. Taking (\ref{f3.30}) into account gives
for large values of $n,$

\begin{equation}
0\leq 1-(1+0_{p}(1))W(x)~/~W(x_{m})
\label{f3.42}
\end{equation}

\begin{equation}
\leq 2\left\{
(z_{m}-x_{m})/(m^{2v}R(x_{m})\right\} ^{2}+3m^{-(\alpha -\beta )},
\label{f3.42}
\end{equation}

\noindent which in turn implies

\begin{equation}
\lim_{x\rightarrow y_{0}}W(x)~/~W(x_{m})=1~in~probability.  \label{f3.43}
\end{equation}

\noindent It follows that

\begin{equation}
\lim_{x\rightarrow y_{0}}W(x)~=\lim_{n\rightarrow +\infty
}~W(x_{n}) \text{ in probability}.  \label{f3.43b}
\end{equation}

\noindent Now Formulas (\ref{f3.39a}) and (\ref{f3.43b}) together give (\ref{f3.30})
which, combined with Lemma \ref{lemmao4} proves Theorem \ref{theo23}.

\section{Concluding comments} \label{sec5}

\subsection{Conjecture} We conjecture that the couple $(A_{n},T_{n})$ should suffice to characterize the 
whole extremal domain, in particular that of the Gumbel subdomain, following the de Haan's functional characterization od $D(\Lambda)$ and $D(\psi)$ 
(Theorems 2.5.6 and 2.6.1 in \cite{dehaan} as reminded in Lemma \ref{lemmao2}). As well, it must be expected, unless a counterexample is given, that the couple of Diop and Lo statistics, given in Remark \ref{remdioplo}, would also be an ECSFEXT.

\subsection{Technical improvements}
The restriction $\beta >\frac{1}{2}$ is required just for (\ref{f3.1}). It is
easily showed that when $F\in \Gamma $, one has lim$_{u\rightarrow
0}(1-G(G^{-1}(1-u))/u=1.$ This remarks remove the conditions $\beta >\frac{1%
}{2}$ in Theorem 1. For the weak limit, it will be shown in the coming paper
that Theorem 1 holds for all sequences $k$ and $\ell $ whenever $%
k/n\rightarrow 0,\quad \ell /n\rightarrow 0,$ $\ell /k^{\frac{1}{2}-\eta
}\rightarrow 0,$ $as$ \ $n\rightarrow +\infty ,$ for some $\eta ,\quad 0\leq
\eta <\frac{1}{2}\left( \eta =0\text{ for }F\in D(\psi )\cup D(\phi )\right)$%

\subsection{Multivariate Gaussian Law of the ECSFEXT}
The existence of family of statistics characterizing some class of
distribution must yield statistical tests. The first step to this is the
determination of the limit laws of the ECSFEXT. This is done in a coming paper.

\section{Technicals Lemmass} \label{sectools}

We invite the reader to remind the definitions of the two first asymptotic moments of a distribution function
in \ref{fm1} and \ref{fm2}. We have the following properties. 

\begin{lemma} \label{lemmao1} For any $\gamma ,$\ $0<\gamma <+\infty$,\\

\noindent \textbf{(i)} $\ F\in D(\Lambda )$ iff $G\in D(\Lambda )$ and $%
R(x,G)\rightarrow 0$ as $x\rightarrow x_{0}(G)=y_{0}$; \newline

\noindent \textbf{(ii)} $F\in D(\Phi _{\gamma })$ iff $G\in D(\Lambda )$ and $%
R(x,G)\rightarrow 1/\gamma $ as $x\rightarrow x_{0}(G)=y_{0}$;\newline

\bigskip 

\noindent \textbf{(iii) }$F\in D(\psi _{\gamma })$ iff $G\in D(\psi _{\gamma })$ and
then $R(x,G)\rightarrow 0$ as $x\rightarrow y_{0}<+\infty$;\newline

\bigskip 

\noindent \textbf{(iv)}$F\in D(\psi _{\gamma })$ iff $F\left( x_{0}-1/.\right) \in
D\left( \psi _{\gamma }\right) .$
\end{lemma}

\begin{proof}
See Lemmas 9 and 10 in L\^{o} (1986) for (i) and Lemma 1 in Mason (1982) for
(ii). Point (iii) is proved similarly to (i) and (ii). (iv) is Part (ii) of
Theorem A. of de Haan (1970).
\end{proof}

\bigskip

\begin{lemma} \label{lemmao2} (de Haan \cite{dehaan}, Theorems 2.5.6 and 2.6.1). We have

\begin{itemize}
\item[(i)] $F\in D(\Lambda )$ iff $W(x,F)$ and $W(x,F)/R(x,F)^{2}\rightarrow 1$, as $x\rightarrow x_{0}.$

\item[(ii)] \textbf{\ }$\ F\in D(\psi _{\gamma })$ iff $x_{0}(F)+\infty $
and $W(x,F)/R(x,F)^{2}\rightarrow (1-1/(2+\gamma )),$as $x\rightarrow x_{0}.$
\end{itemize}
\end{lemma}

\noindent Functions $s(u),0<u<1,$ such that for all $\lambda >0$, 
\begin{equation*}
\lim_{u\rightarrow 0}s(\lambda u)/s(u)=1,
\end{equation*}%
are called Slowly Varying functions at Zero (SVZ) and are greatly involed in
our proofs. We recall here some of their properties before we state some
basic results of $df$'s lying in $\Gamma $.

\newpage
\begin{lemma}
\label{lemmao3} Let $s(u),0<u<1,$ be SVZ. Then,

\begin{itemize}
\item[(i)] It admits the Kamarata's representation (KARARE) :

\begin{equation}
s(u)=c(u)\exp \left( \int_{u}^{1}b(t)t^{-1}dt\right) ,\text{ \ \ \ }0<u<1
\label{f2.1}
\end{equation}%
where $c(u)\rightarrow c$, $0<c<+\infty $, $b(u)\rightarrow 0$, as $%
u\rightarrow 0$.

\item[(ii)] For any $\delta >0$, $u^{-\delta }s(au)\rightarrow 0$, as $%
au\rightarrow 0$ and $a\times u\rightarrow 0.$
\end{itemize}
\end{lemma}

\begin{proof}
(\ref{f2.1}) in (i) is well-known. See Lemma 12 of L\^{o} (1986a) (\cite%
{gsloa})for (ii). Anyway it is easily derived from (\ref{f2.1}.
\end{proof}

\begin{lemma} \label{lemmao4} We have

\begin{itemize}
\item[(i)] $F\in D(\Phi _{\gamma })$ $iff$ $u^{1/^{\gamma }}F^{-1}(1-u)$ $is$
$SVZ $

\item[(ii)] $F\in D(\psi _{\gamma })$ $iff$ $u^{1/^{\gamma
}}(x_{0}(F)-F^{-1}(1-u)) $ $is$ $SVZ$

\item[(iii)] $F\in D(\Lambda )$ $iff$ $there$ $exists$ $a$ $SVZ$ $function$ $%
s(u),$ $0\leq u\leq 1$ and a constant $b$ $such$ that$\bigskip $%
\begin{equation}
F^{-1}(1-u)=b-s(u)+\int_{u}^{1}s(t)t^{-1}dt,0<u<1  \label{f23}
\end{equation}

\noindent When \ref{f23}) holds, one may take 
\begin{equation}
s(u)=u^{-1}\int_{1-u}^{1}(1-s)dF^{-1}(s)=:r(u,F).  \label{f23a}
\end{equation}
\end{itemize}
\end{lemma}

\begin{proof}
See Theorem 2.4.1 of De Haan (1970) for (i). (ii) is easily derived from (i)
and Part (iii) of theorem A. \ The representation \ref{f23} is due to De Haan (1970) \cite{dehaan}. The writing of $s(u)$ as (\ref{f23a}) is due to
Deheuvel-Haeusler-Mason (1988) \cite{dhm}.
\end{proof}

\begin{lemma} \label{lemmao5} Let $G$ be any distribution function. Then \newline
\textbf{(i)} for all $0<u<1,$ \ \ $G(G^{-1}(u)=u$ \ or $u$ lies on
a constancy interval of $G^{-1}$.\newline
\noindent \textbf{(ii)} for all $-\infty <x<+\infty ,G^{-1}(G(x))=x$ or x
lies on a constancy interval of $G.$
\end{lemma}

\begin{proof}
Notice that the set of discontininuity points of $G,$ say $D,$ is countable.
And for all $x \in D,\left[ G(x-),G(x)\right] =:\left[ v_{x},u_{x}\right]
$ \ is a constancy interval of $G^{-1}$ with

\begin{equation}
\forall u\in \left[ v_{x},u_{x}\right] ,G^{-1}(u)=x  \label{f24}
\end{equation}

\noindent Now let $I=U_{x\in D}$ $\left[ v_{x},u_{x}\right[ .$ One has

\begin{equation}
\forall u\in I,\text{ \ \ \ \ }G(G^{-1}(u))>u  \label{f25a}
\end{equation}

\noindent One also has

\begin{equation}
\forall x\in D,\text{ \ \ \ \ \ }G(G^{-1}(u_{x}))=u_{x}  \label{f25b}
\end{equation}

\noindent It follows that the complementary $J$ of $I$ in $(0,1)$ is not empty. To
finish, we have to show

\begin{equation}
(u\in J\text{ }and\text{ }x=G^{-1}(u))\Rightarrow \text{ \ }(G(x)=u).
\label{f26}
\end{equation}

\noindent Suppose that for some $u\in J,x=G^{-1}(u)$ and $G(x)>u$. Thus\\

\noindent \textbf{either} , $x$ is a continuity point and there exists a sequence $%
x_{n}\uparrow x$ such that $G(x_{n})\uparrow G(x)$ $as$ $n\rightarrow
+\infty .$ Hence for some $\eta ,n>$ $\eta ,x_{n}>u$ and $G(x_{n})>u$ so
that $G^{-1}(u)<x,$ which leads to a contradiction ;\\

\noindent \textbf{or} $x$ is a discontinuity point and thus $u$ lies on $\left[ v_{x},\text{ \ }%
u_{x}\right[ ,$ which is a constancy interval of $G^{-1}$ and, by
consequence, $x\in I.$ This also leads to a contradiction. These two
contradictions imply \ref{f26} which, combined with\ref{f25a}, prove
Part (i).
\end{proof}

\begin{lemma} \label{lemmao6} Let $G \in D(\Lambda)$ then \\

\noindent \textbf{(i)} $(1-G(G^{-1}(1-u))/u\rightarrow 1$ as $u\rightarrow 0$ \\

\noindent and\\

\noindent \textbf{(ii)} $R(G^{-1}(1-u))$ is SVZ.
\end{lemma}

\begin{proof}
\textbf{Proof of Part i}. \textbf{Either} $(1-G(G^{-1}(1-u))=1-u$ and thus

\begin{equation}
(1-G(G^{-1}(1-u))u=1;  \label{f27a}
\end{equation}

\noindent \textbf{Or}, by Lemma \ref{lemmao5}, $1-u \in \left[ G(x-),G(x)%
\right[$, for some discontinuity point $x,$ with $1-G(x)<u\leq 1-G(x-)$ and
hence

\begin{equation}
1-G(x)=(1+c(x))\exp \left( \int_{-\infty }^{x}\Phi (t)^{-1}dt\right)
,-\infty <x<x_{0}(G)=y_{o},  \label{f27b}
\end{equation}

\noindent where $c(x)\rightarrow 0,$ \ \ \ $\Phi ^{\prime }(x)$ exists and $%
\Phi ^{^{\prime }}(x)\rightarrow 0$ as $x\rightarrow y_{0},$ yields

\begin{equation}
\left( 1-G(x)\right) /(1-G(x-))=(1+c(x))/(1+c(x))/(1+c(x-1))\rightarrow 1%
\text{ }as\text{ }x\rightarrow y_{0}  \label{f27d}
\end{equation}

\noindent This together with (\ref{f27a}), and (\ref{f27b}) prove Part
i).\newline

\bigskip

\noindent \textbf{Proof of Part ii}. From (\ref{f23}), one has \ (cf.
Lemma 4 in L\^{o} (1989)),

\begin{equation}
\left\{ G^{-1}(1-\lambda u)-G^{-1}(1-u)\right\} /s(u)\rightarrow -\log
\lambda ,\text{ }as\text{ }u\rightarrow 0\text{ ;}  \label{f28a}
\end{equation}

\noindent Thus

\begin{equation}
R(G^{-1}(1-u)\sim s(u)\text{ }as\text{ }u\rightarrow 0,  \label{f28c}
\end{equation}

\noindent which proves Part ii) since $s(u)$ is SVZ.
\end{proof}

\noindent We now introduce two useful and important lemmas.

\begin{lemma}
\label{lemmao7}. Let $F$ be any distribution function satisfying

\begin{itemize}
\item[(i)] $R(x)$ and $W(x)$ are finite for $x<x_{o}(F)$ ;

\item[(ii)] $(z-x)/R(z)\rightarrow +\infty ,$ as $x\rightarrow x_{0},$ \ \ $%
x<z$ $;$

\item[(iii)] $(z-x)^{2}/W(x)\rightarrow +\infty ,$ as $x\rightarrow x_{0},$ $%
\ \ z\rightarrow x_{0},$ $\ x<z;$
\end{itemize}

\noindent Then, $R(x,z)=R(x)\rightarrow 1$ $and$ $W(x,z)/W(x)\rightarrow 1$ $%
as$ $x\rightarrow x_{0},$ \ $z\rightarrow x_{0},$ \ $x<z.$
\end{lemma}

\begin{proof}
\textbf{Point (i)}. We have, for $x<z,$

\begin{equation}
R(x,z)\text{ }=\text{ }R(x)(1\left\{ \int_{z}^{x_{0}}1-F(t)dt\right\}
/\left\{ \int_{z}^{x_{0}}1-F(t)dt\right\} )=:R(x)(1-E_{0})  \label{f2.9}
\end{equation}

\bigskip 

\noindent Since $(1-F(t))/(1-F(1-F(z))\geq 1$ when $x\leq t\leq z,$ we get

\begin{equation}
0\leq E_{1}=(1+\left\{ \int_{x}^{z_{0}}1-F(t)dt\text{ }dy\right\} \text{ }/%
\text{ }\left\{ \int_{z}^{x_{0}}\int_{x}^{z_{0}}1-F(t)dt\text{ }dy\right\}
)^{-1}  \label{f2.10}
\end{equation}

\noindent \textbf{Point (ii)}. Formula (\ref{f2.10}) and Assumption (ii) allow
to conclude that

\bigskip 
\begin{equation}
R(x,z)=R(x)\rightarrow 1\text{ }as\text{ }x\rightarrow x_{0},\text{ \ \ }%
z\rightarrow x_{0},\text{ \ \ }x<z.  \label{f2.11}
\end{equation}

\noindent \textbf{Point (iii)}. We have

$$
W(x,z) =W(x)\left( 1-(\left\{ \int_{z}^{x_{0}}\int_{y}^{x_{0}}1-F(t)\text{ 
}dt\text{ }dy\right\} \text{ }/\text{ }\left\{
\int_{z}^{x_{0}}\int_{y}^{x_{0}}1-F(t)\text{ }dt\text{ }dy\right\} \right)
$$

\begin{equation}
-\left( \left\{ \int_{x}^{z}\int_{z}^{x_{0}}1-F(t)\text{ }dt\text{ }%
dy\right\{ \text{ }/\text{ }\left\{ \int_{z}^{x_{0}}\int_{y}^{x_{0}}1-F(t)%
\text{ }dt\text{ }dy\right\} \right) \label{f2.12}
\end{equation}

$$
=W(x)(1-E_{1}-E_{2})
$$

\noindent We approximate $E_{1}$ and $E_{2}$. First, the inequality $(1-F(t)\geq (1-F(z))$ for $x\leq t\leq z$ and Assumption
(iii) yield

\begin{equation}
0\leq E_{1}=(1+(z-x)^{2}/W(z))^{-1}\rightarrow 0\text{ }as\text{ }%
x\rightarrow x_{0},\text{ \ \ \ \ }y\rightarrow x_{0},\text{ } x<z.
\label{f2.14}
\end{equation}

\noindent Secondly,

\begin{equation}
0\leq
E_{2}=(z-x)(\int_{z}^{x_{0}}1-F(t)dt/\int_{z}^{x_{0}}%
\int_{x}^{z_{0}}1-F(t)dtdy  \label{f2.15}
\end{equation}

\begin{eqnarray*}
&=&(z-x)\frac{R(z)}{W(z)}E_{1}\leq \frac{R(z)}{z-x}\text{ }\frac{(z-x)^{2}}{%
w(z)}(1+(z-x)^{2}/W(z))^{-1} \\
&\leq &\sup \text{ }\underset{x\in R^{+}}{}\left| (x(1+x)^{-1}\right| \text{ 
}\frac{R(z)}{z-x}\leq R(z)/(z-x)\rightarrow 0,
\end{eqnarray*}

\noindent as $x\rightarrow x_{0}$, $z\rightarrow x_{0}$, $x<z$, by
Assumption (i). Statements \ref{f2.12}, \ref{f2.14} and \ref{f2.15} yield

\begin{equation}
W(x,z)\text{ }/W(x)\rightarrow 1\text{ }as\text{ }x\rightarrow x_{0},\text{
\ }z\rightarrow x_{0},\text{ \ }x<z.  \label{f2.17}
\end{equation}

\bigskip

\noindent Now Formulas \ref{f2.11} and \ref{f2.17} together prove Lemma \ref{lemmao7}.
\end{proof}

\begin{lemma}
\label{lemmao8} Let $F\in \Gamma $, then for $u=1-G(x),$ $v=1-G(z),$ \ \ \ $%
v/u\rightarrow 0$ as \ $x\rightarrow y_{0},$ \ $z\rightarrow y_{0}.$
\end{lemma}

\begin{proof}
Let $F\in \Gamma$. Either $G\in D(\Lambda )$ or $G\in D(\psi _{\gamma }),$ \ 
$\gamma >0.$ \ It suffices to prove Part (i). Part (ii) will follow from
Part (i) and Lemma \ref{lemmao2}. If $G\in D(\Lambda )$ , Lemma 1 in L\^{o}
(1986a) implies Part (i).\\

\noindent If $G\in D(\psi _{\gamma })$ , $G(y_{0}-1/^{.})\in D(\phi _{\gamma
})$ and \ thus $(y_{0}-y)^{-\gamma }(1-G(y))$ is SV at infinity. KARARE
yields

\begin{equation}
1-G(y)=c(y)(y_{0}-y)^{\gamma }\exp
(\int_{y_{1}}^{(y_{0}-y)^{-1}}b(t)^{-1}dt),(y_{0}-y)^{-1>}y_{1},
\label{f2.18}
\end{equation}

\noindent where $c(y)\rightarrow c,$ \ $0<c<+\infty ,$ \ $as$ \ $y\rightarrow y_{0}$
and $b(t)\rightarrow 0$ $as$ $t\rightarrow +\infty .$ for $\eta $ such that

\begin{equation}
\sup_{t\geq n}b(t)\leq \varepsilon <\gamma ,\text{ \ \ \ \ \ \ }y_{0}-\eta
^{-1}\leq y<y_{o}  \label{f(2.19)}
\end{equation}

\noindent Formula (\ref{f2.18}) implies

\begin{equation}
C_{1}.((y_{0}-x)/(y_{0}-z))^{\gamma +\varepsilon }\leq
(1-G(x))/(1-G(z))=u/v\leq C_{2}.((y_{0}-x)/(y_{0}-z))^{\gamma -\varepsilon }
\label{f(2.20)}
\end{equation}

\noindent where $C_{1}$ and $C_{2}$ are positive constants. The right
inequality ensures that

\begin{equation}
(y_{0}-x)/(y_{0}-z)\rightarrow +\text{ }\infty ,\text{ }as\text{ }%
u\rightarrow 0,\text{ \ }v/u\rightarrow 0.  \label{f(2.21)}
\end{equation}

\noindent By using now Formula 2.6.4 of De Haan \cite{dehaan}, we get for some constant $C$,

$$
(z-x)/R(z)\sim C \text{ \ \ }(z-x)/(y_{0}-z)
$$

\begin{equation}
Const.\text{ \ }%
(-1+(y_{0}-x)/(y_{0}-z))\rightarrow +\infty  \label{f(2.22)}
\end{equation}

as $x\rightarrow y_{0}$, $z\rightarrow y_{0}$, $(1-G(z))/(1-G(x))\rightarrow 0$.\\

\noindent Part (i) is now proved. Part (ii) follows from Lemma \ref{lemmao2}.
\end{proof}

\noindent To finish with this section, we recall properties of empirical
distribution functions ($edf$). The $edf$ associated with $Y_{1},...,Y_{n}$ is
defined by

\begin{equation}
G_{n}(x)= \# \left\{ i,\text{ \ \ }1\leq i\leq n,\text{ \ \ \ }Y_{i}\leq
x\right\} /n,\text{ \ \ \ \ }x\in R  \label{ref(2.23)}
\end{equation}

\noindent Let $U_{n}(s),$ \ 0$\leq s\leq 1$, be the $edf$ associated with $%
U_{1},...,U_{n},$a.s.i.c. of a uniform $rv$ on $(0,1)$.\\

\noindent \textbf{Fact 1.} We may WLOG and do assume that

\begin{equation*}
\left\{ 1-G_{n}(x),\text{ \ \ }x\in R,\text{ \ \ \ }n\geq 1\right\} =\left\{
U_{n}(1-G(x)),\text{ \ }x\in R,\text{ \ \ }n\geq 1\right\}
\end{equation*}

\bigskip 

\noindent \textbf{Fact2}. For all $\mu$, $0< \mu \frac{1}{2}$, 
\begin{equation*}
\limsup_{n\rightarrow +\infty } n^{\mu} \sup_{0\leq s\leq 1} \left|
U_{n}(s)-s\right| <\infty. \: a.s 
\end{equation*}

\noindent \textbf{Fact 3} Let $k=\left[ n^{\alpha }\right] ,\frac{1}{2}%
<\alpha <1.$ Then $kU_{k,n}/n\rightarrow 1,$ $a.$ $s.$ as $n\rightarrow
+\infty$.\newline

\bigskip

\noindent \textbf{Fact 4}. We have

\begin{equation*}
\lim_{\lambda \uparrow +\infty } \lim \inf_{n\uparrow +\infty }P(\lambda
^{-1}\leq \inf_{U_{1,n}\leq s\leq 1} U_{n}(s)/s\leq \sup_{U_{1,n}\leq s\leq
1} U_{n}(s)/s\leq \lambda )=1. 
\end{equation*}

\bigskip

\noindent Fact 2 is derived from Theorem 4.5.2 and Formula 1.2.3 both in M.
Cs\"{o}rgo-R\'{e}v\'{e}z (1981). Fact 3 is a consequence of Fact 2. Fact 4
is quoted in S. Cs\"{o}rgo and al.(1985) for the quantile process. A simple
change of variable suffices to put it in the form of Fact 4.\newline

\noindent We now introduce a general device which permits to overcome
discontinuity problems.\newline

\bigskip

\noindent \textbf{Fact 5}. For $n$ fixed, there exists a sequence $\left(
t_{p}\right) _{p\geq 1}$ such that $t_{p}\uparrow Y_{n,n}$ $as$ $p\uparrow
+\infty $ and for all $p\geq 1$,

\begin{equation*}
\inf_{U_{1,n}\leq s\leq 1}U_{n}(s))/s\leq \underset{0\leq x\leq t_{p}}{\inf }
\frac{U_{n}(1-G(x))}{1-G(x)}\leq \underset{0\leq x\leq t_{p}}{\sup } \frac{%
U_{n}(1-G(x))}{1-G(x))}\leq \sup_{U_{1,n}\leq s\leq 1}U_{n}(s)/s. 
\end{equation*}

\noindent \textbf{Proof}. By using the representations of the constancy
intervals of $G^{-1}$ given in the proof of Lemma \ref{lemmao5}, we remark
that :\newline

\bigskip

\noindent \textbf{(i) either} $G(Y_{n,n})=G(G^{-1}(1-U_{1,n}))=1-U_{1,n}$ and
it suffices to put $t_{p}=Y_{n,n}$ for all $p\geq 1$;\newline

\bigskip

\noindent \textbf{(ii) or} $G^{-1}(1-U_{1,n}))>1-U_{1,n}$ and, necessarily, $%
1-U_{1,n}$ lies in some constancy interval of $G^{-1},\left] v_{j},u_{j}%
\right] .$ Putting $t_{p}=Y_{n,n}-1/p$ for $p\geq 1,$ one has $G(t_{p})\leq
v_{j}\leq 1-U_{1,n}$ for all $p\geq 1.$ Thus for $0\leq x\leq t_{p},$ \ $%
1-G(x)\geq U_{1,n},$ for alla $p\geq 1.$ This completes the proof.

\section{APPENDIX} \label{secapp}

\subsection{A counterexample about the de Haan-Resnick} \label{subsecapp1}

We show here that the De Haan-Resnick estimator for the index of a stable
law defined by

\begin{equation}
C_{n}=\frac{Y_{n,n}-Y_{n-k,n}}{\log k} \label{haanresnick}
\end{equation}

\noindent does not characterize $D(\phi_{\gamma})$ as Hill's estimator does. To prove this, we begin to remark, as Mason (1982) (\cite{mason}, (cf. its appendix) showed it, that the $df$ $G$ defined by

\begin{equation}
G^{-1}(1-2^{-m})=m,~m=0,1,2...
\end{equation}

\noindent and

\begin{equation*}
G^{-1}(1-u)=m+(2^{-m}-u)(2^{m+1}), \: if~2^{-m-1}<u<2^{-m},
\end{equation*}

\noindent does not belong to $D(\phi)$. Thus by Lemma \ref{lemmao1}, $F(\cdot)=G(log(\cdot))$ does not belong to $\Lambda$. However, it is easy to check
that

\begin{equation}
\frac{G^{-1}(1-s)-G^{-1}(1-bs)}{\log ~b}\rightarrow (\log ~2)^{-1}, \text{ }
as \text{ } s\rightarrow 0, \text{ } b\rightarrow +\infty \text{ } and \text{
} bs\rightarrow 0
\end{equation}

\noindent Letting $b=U_{k,n}/U_{1,n}$, we get, via representation (\ref{f3.0}%
),

\begin{equation}
C_{n}\rightarrow _{p}(Log2)^{-1}, \text{ } as \text{ } n\rightarrow +\infty
\end{equation}

\noindent Thus the convergence of $C_{n}$ to a positive and finite real
number for all sequences $k\rightarrow +\infty $ verifying $k/n\rightarrow
0~as~n\rightarrow \infty $, does not imply that $F$ belongs to $D(\phi ).$

\subsection{A useful identity that links Lo and Dekkers \textit{et al.} estimators} \label{subsecapp2}

We prove here the following identity in the following

\begin{lemma}
\bigskip Let $1\leq k\leq n$ be integers and $x_{k},...,x_{n}$ $(n-k+1)$
real numbers. The we have
\begin{equation}
\sum_{i=1}^{k}(x_{n-i+1}-x_{n-k})^{2}=2 \sum_{i=1}^{k}\sum_{j=1}^{i}j(1-\delta
_{ij}/2)(x_{n-i+1}-x_{n-i})(x_{n-j+1}-x_{n-j});  \label{ilo}
\end{equation}
where $\delta _{ij}=1$ if $i=j$ and $0$ elsewhere, is the Kronecker symbol.
\end{lemma}

\bigskip 

\begin{proof}
We use these notations $S_{r}=\sum_{1\leq i\leq k}x_{n-j+1}^{r},$ $r=1,2.$
We also these two formulas :
\begin{equation}
\sum_{i=1}^{h}j(x_{n-j+1}-x_{n-j})=x_{n}+...+x_{n-h+1}-hx_{n-h}  \label{id1}
\end{equation}
and
\begin{equation*}
\sum_{i=1}^{k}\sum_{j=1}^{i}j(x_{n-i+1}-x_{n-i})(x_{n-j+1}-x_{n-j})=\frac{1}{%
2}\left\{ (\sum_{j=1}^{k}x_{j})^{2}-\sum_{j=1}^{k}x_{j}^{2}\right\}
\end{equation*}

\begin{equation}
=(S_{1}^{2}-S_{2})/2.  \label{id2}
\end{equation}

\noindent The first of the formulas is obtained by induction and proved for h=2, 3,
etc.. and then, the deduction is easy to get. The second is simply deduced
for the developpement of the square of the sum of the $n-k+1$ numbers. Now, the
second term of (\ref{ilo}) is
\begin{equation*}
\sum_{i=1}^{k}\sum_{j=1}^{i-1}j(x_{n-i+1}-x_{n-i})(x_{n-j+1}-x_{n-j})+\frac{1%
}{2}\sum_{i=1}^{k}i(x_{n-i+1}-x_{n-i})^{2}\equiv A+B.
\end{equation*}
Next, using (\ref{id1}), on has 
\begin{equation*}
A=\sum_{i=1}^{k}(x_{n-i+1}-x_{n-i})\sum_{j=1}^{i-1}j(x_{n-j+1}-x_{n-j})
\end{equation*}

\begin{equation*}
=\sum_{i=1}^{k}(x_{n-i+1}-x_{n-i})(x_{n}+...+x_{n-i+2}-(i-1)x_{n-i+1})
\end{equation*}

\begin{equation*}
=\sum_{i=1}^{k}(x_{n-i+1}-x_{n-i})(x_{n}+...+x_{n-i+1}-ix_{n-i+1})
\end{equation*}
\begin{equation*}
=\sum_{i=1}^{k}\sum_{j=}^{i}x_{n-i+1}x_{n-j+1}-\sum_{i=1}^{k}ix_{n-i+1}^{2}-%
\sum_{i=1}^{k}\sum_{j=1}^{i}x_{n-i}x_{n-j+1}+\sum_{i=1}^{k}ix_{n-i}x_{n-i+1}
\end{equation*}
\begin{equation*}
\equiv A_{11}+A_{12}+A_{21}+A_{22}.
\end{equation*}
Now by the change of variables $i=h-1$%
\begin{equation*}
-A_{21}=\sum_{h=2}^{k+1}\sum_{j=1}^{h-1}x_{n-h+1}x_{n-j+1}=\sum_{h=1}^{k}%
\sum_{j=1}^{h-1}x_{n-h+1}x_{n-j+1}+\sum_{j=1}^{k}x_{n-j+1}x_{n-k}
\end{equation*}
\begin{equation*}
=\sum_{h=1}^{k}\sum_{j=1}^{h}x_{n-h+1}x_{n-j+1}-\sum_{h=1}^{k}x_{n-j+1}^{2}+%
\sum_{j=1}^{k}x_{n-j+1}x_{n-k}
\end{equation*}
\begin{equation*}
=A_{11}-S_{2}+x_{n-k}S_{1}.
\end{equation*}
Further
\begin{equation*}
2B=\sum_{i=1}^{k}ix_{n-i+1}^{2}+\sum_{i=1}^{k}ix_{n-i}^{2}-2%
\sum_{i=1}^{k}ix_{n-i}x_{n-i+1}=B_{1}+B_{2}+B_{3}
\end{equation*}
with, by change of variable \ $i=h-1$,
\begin{equation*}
B_{2}=\sum_{h=2}^{k+1}(h-1)x_{n-h+1}^{2}=\sum_{h=1}^{k+1}(h-1)x_{n-h+1}^{2}=%
\sum_{h=1}^{k}(h-1)x_{n-h+1}^{2}+kx_{n-k}^{2}
\end{equation*}
\begin{equation*}
=\sum_{h=1}^{k}hx_{n-h+1}^{2}-\sum_{h=1}^{k}x_{n-h+1}^{2}+kx_{n-k}^{2}=%
\sum_{h=1}^{k}hx_{n-h+1}^{2}-S_{2}+kx_{n-k}^{2}.
\end{equation*}
Finally
\begin{equation*}
B=\frac{1}{2}(-A_{12}-A_{21})-\frac{1}{2}S_{2}+\frac{1}{2}kx_{n-k}^{2}-A_{22}
\end{equation*}
and the second term of \ref{ilo} is 
\begin{equation*}
A_{11}+A_{12}-A_{11}+S_{2}-x_{n-k}S_{1}+A_{22}-A_{12}-\frac{1}{2}S_{2}+\frac{%
1}{2}kx_{n-k}^{2}-A_{22}
\end{equation*}
\begin{equation*}
=\frac{S_{2}-2x_{n-k}S_{1}+kx_{n-k}^{2}}{2}.
\end{equation*}
This is nothing than the half of
\begin{equation*}
\sum_{i=1}^{k}(x_{n-i+1}-x_{n-k})^{2}=\sum_{i=1}^{k}x_{n-i+1}^{2}-2x_{n-k}%
\sum_{i=1}^{k}x_{n-i+1}+kx_{n-k}^{1}=S_{2}-2x_{n-k}S_{1}+kx_{n-k^{2}}.
\end{equation*}
This achieves the proof.
\end{proof}

\end{document}